\documentclass[11pt]{article}
\usepackage{authblk}
\usepackage{amssymb}
\usepackage{amsmath}
\usepackage{algorithm}

\newenvironment{proof}{\noindent {\bf Proof }}
{\hfill $\bullet$ \vspace{0.25cm}}

\def\one{{\bf 1}\hskip-.5mm}

\def\E{{\mathbb E}}

\def\R{{\mathbb R}}
\def\Z{{\mathbb Z}}

\def\N{{\mathbb N}}
\def\GG{{\mathcal G}}
\def\SS{{ \mathcal S}}
\def\PP{{ \cal P}}
\def\F {{\mathcal F}}

\def\LL{{\mathcal L}}
\def\II{{\mathcal I}}
\def\B{{\mathcal B}}

\newtheorem{theo}{Theorem}
\newtheorem{prop}{\indent Proposition}
\newtheorem{lem}{\indent Lemma}

\newtheorem{ex}{\indent Example}

\begin{document}

\title{A stochastic system with infinite interacting components to model the time evolution of the membrane potentials of a population of neurons}

\author{K. Yaginuma}

\affil{
              Instituto de Matem\'atica e Estat\'istica \\
              Universidade de S\~ao Paulo \\
             karinayy@ime.usp.br             
}
\date{}

\maketitle

\begin{abstract}
We consider a new class of interacting particle systems with a countable number of interacting components. The system represents the time evolution of the membrane potentials of an infinite set of interacting neurons. We prove the existence and uniqueness of the process, using a perfect simulation procedure. We show that this algorithm is successful, that is, we show that the number of steps of the algorithm is almost surely finite. We also construct a perfect simulation procedure for the coupling of a process with a finite number of neurons and the process with a infinite number of neurons. As a consequence, we obtain an upper bound for the error that we make when sampling from a finite set of neurons instead of the infinite set of neurons.

\end{abstract}

\section{Introduction}
\label{intro}
The activity of a neuron is manifested by the emission of \textit{action potentials} (spikes). Such activity is performed by each neuron in a system composed of a large number of interacting neurons. An estimate of the number of neurons in the human brain is about 80 billion. Moreover, it is estimated that a single neuron can make about 10 thousand distinct connections with other neurons in the system. While the shape of the action potential is essentially constant for a given neuron in such a way that each neuron can be distinguished by the shape of its action potential, the spike train generated by a neuron depends on its intrinsic characteristics as well as on stimuli from other neurons and external to the system.

The action potential occurs with a probability that is an increasing function of the concentration of sodium ions inside the cell, which enter the cell via voltage-gated sodium channels that are transiently open when the cell is excited. At the end of this process potassium channels are activated and produce an outward potassium current that restores the membrane potential to resting level. Once an action potential is generated, it propagates through the axon of the neuron and when it reaches the axon terminal neurotransmitters are released into the space between the neuron and a nearby neuron (synaptic cleft). These neurotransmitters diffuse through the synaptic cleft and some of them bind to synaptic receptors in the membrane of the other neuron and change its membrane potential, either exciting or inhibiting the neuron, for more details see \cite{gerstner-2002}.

This is the motivation for the introduction of a new class of interacting particle systems with infinite range that we consider in this paper. To deal with the problem of the large number of neurons seems natural to consider a system with a countable set of components with interactions of infinite range. The system describes the evolution of the membrane potentials of neurons in continuous time. 

This class of systems is based on the class of models introduced by Galves and L\"ocherbach \cite{galves-eva-2013}. They consider a system in which the probability of a spike of a given neuron depends on the accumulated activity of the system after its last spike. The temporal evolution of the system occurs in discrete time. For each neuron at each time the symbol 1 is assigned if there is a spike at that time or 0 otherwise. If one looks at a single neuron, the temporal evolution seems like a stochastic chain with memory of variable length where the neuron needs to look to the past until it finds its last spike. 

In \cite{bruno-2011}, Cessac suggested the same kind of dependence from the past in the framework of leaky integrate and fire models. He considered a system with a finite number of membrane potential processes in discrete time. The image of this process which can be described as a spike train is also a stochastic chain of infinite order where each neuron has to look back into the past until its last spike time. The process described by Cessac is a finite version of the model considered in \cite{galves-eva-2013}. 

In our model, we do not need to look back into the past in order to determine the probability of a spike, the information about the accumulated activity is in the membrane potential of the neurons. We consider an interacting particle system with a countable number of components having natural state space and interactions of infinite range.   

This paper is organized as follows. In Section \ref{definition} we present the model and the two main results. The first theorem proves the existence and uniqueness of the model. In the second theorem, we give an upper bound to the probability of sampling different outputs of a perfect simulation algorithm based on a coupled construction between a process considering a finite set of neurons and a process considering the whole system of neurons. Our main technical tool is a perfect simulation procedure for the processes. Perfect simulation is the name given to any algorithm that generates as output a sample of stationary processes whose distribution is ensured that follows a given law of probability. It can be shown that when such a sample can be produced the process exists and is unique by construction.  

The perfect simulation scheme that we used is the \textit{clan of ancestors}, introduced in \cite{ferrari-2002}. The algorithm is based on a two-step procedure: (i) in the first step we define a scheme of perfect simulation to generate the set of \textit{ancestors}, which are the predecessors that may have an influence on the target process, and (ii) in the second step, using the information of the first step, the sample is generated according
to the interaction rules of the target process. The second step, and hence the whole procedure, is feasible if these \textit{ancestors} form a finite set with probability one. 

Based in the clan of ancestor scheme, in \cite{galves-2010} they address the problem of perfect simulation for the Gibbs measure with infinite range interactions. In \cite{galves-nancy-2013}, they consider a particle system in $\Z^d$ with real state space and interactions of infinite range. Assuming that rates of change are continuous, they obtained a Kalikow decomposition of the rates with infinite range. Thus, rates are represented by a mixing of finite range rate. As an application of the decomposition they introduced a perfect simulation algorithm for the system. The same technique is used in \cite{galves-eva-2013}, where they introduced a Kalikow decomposition for the transitions probabilities with infinite order. In our case, we do not need to use the technique of Kalikow decomposition, since the rates of the system does not depend on the global configuration of the system. The technique used here is inspired by the construction of a system of interacting Markov processes presented in \cite{galves-1977}. 

In Section \ref{backward}, we introduce the backward sketch process which is the basis of the perfect simulation algorithm. Section \ref{algorithm} presents the coupled perfect simulation algorithm for the couple processes and presents the results ensures that the algorithm is successful. Sections \ref{proof13} and \ref{proof2} are devoted to the proof of the mains results of the paper.  

\section{Definition and results}\label{definition}
We consider an interacting particle systems on $\II$, a countable set of neurons, having state space given by $\{0,1,\dots\}$ and interactions of infinite range. The elements of the state space are called \textit{potential}. To each neuron in $\II$ we assign a potential. The potential of the neurons changes as times goes by.

We denote by $S = \{0,1,\dots\}^\II$ the configuration space with its product sigma algebra $\SS$. We call an element of $S$ a configuration. Configurations will be denoted by Greek letters $\xi, \eta, \zeta,\dots$. For each neuron $i \in \II$ at time $t \in \R$, $\xi_t(i) \in \N$ represent the membrane potential of $i$ at time $t$. Let $\{V^k_{i \to \cdot}; i \in \II, k\geq 1\}$ be a family of subsets of $\II$, where $V^k_{i \to \cdot}$ is the set of neurons that are affected by the spike of neuron $i$, when the membrane potential of $i$ immediately before the spike is greater or equal to $k$. Note that the sets $V^k_{i \to \cdot}$ doesn't need to be finite.

To define the dynamics of the process $(\xi_t)_{t\in\R}$, we need to introduce the family of transformations $\{\pi^k_i; i\in \II, k\geq0\}$ such that $\pi^k_i:S \to S$ defined as follows, for $i \in \II$ and $k\geq1$,
\begin{eqnarray*}
\pi^{k}_i(\xi)(u) & = &\left\{\begin{array}{ll}
\xi(u) & \hbox{\ if\ } \xi(i) < k, \\
\xi(u) & \hbox{\ if\ } u\neq i, u \not\in V^k_{i \to \cdot} \hbox{\ and\ } \xi(i) \geq k,\\
\xi(u)+1 & \hbox{\ if\ } u \neq i, u \in V^k_{i \to \cdot} \hbox{\ and\ } \xi(i) \geq k, \\
0 & \hbox{\ if\ } u = i \hbox{\ and\ } \xi(i) \geq k.
\end{array}\right.
\end{eqnarray*}

For $k = 0$
\begin{eqnarray*}
\pi^{0}_i(\xi)(u) & =& \left\{\begin{array}{ll}
\xi(u) & \hbox{\ if\ } u\neq i, \\
\xi(u)+1 & \hbox{\ if\ } u = i.
\end{array}\right.
\end{eqnarray*} 

In other words, given a configuration $\xi$ the transformation $\pi^k_i$ only changes the configuration of the system if the membrane potential of $i$ is equal or greater than the value $k$ of the transformation, that is,  if $\xi(i) \geq k$. If $\xi(i) < k$, we say that the neuron $i$ does not have sufficient membrane potential to produce a spike. 

In the case where an spike occurs in $i$, the neuron loses all its membrane potential returning to a resting state, that is represented by $0$. Simultaneously, the membrane potentials of the neurons within the set $V^k_{i \to \cdot}$ are changed by the addition of one unit in membrane potential of each neuron. The external stimulus is represented by the transformations $\pi^0_i$ for $i \in \II$. The external stimulus does not causes an immediate spikes in the neurons but increases the chances of a spike increasing the membrane potential of the neuron.

Now we are ready to introduce the dynamics of our process. We consider markovian particle systems on $S$ whose generator $\LL$ is defined as follows, for any boundary cylinder function $f:S \to \R$ 
\begin{eqnarray}\label{gerador}
\LL f (\xi) = \sum_{i \in \II}\sum_{k\geq 0}\lambda_i(k)[f(\pi^k_i(\xi)) - f(\xi)],
\end{eqnarray}
where $\lambda_i: \N \to \R_+$, for $i \in \II$. 

For each $i \in \II$, define 
\begin{eqnarray}
\Lambda_i = \lambda_i(0) + \sum_{k\geq1}\lambda_i(k) + \sum_{k\geq1}\sum_{j\in V_{\cdot \to i}^k}\lambda_j(k),
\end{eqnarray}
where $V^k_{\cdot \to i} = \{j \in \II : i \in V_{j \to \cdot}^k\}$, and 
\begin{eqnarray}
\rho_i = \dfrac{\lambda_i(0)}{\sum_{k\geq0}\lambda_i(k)}.
\end{eqnarray}

The existence and uniqueness of this class is granted in our first theorem. 

\begin{theo}\label{T1}
Assume that 
\begin{enumerate}
\item \begin{eqnarray}\label{T1H1}
\sup_{i \in \II}\Lambda_i = \beta < +\infty,
\end{eqnarray}
\item For all $i \in \II$
\begin{eqnarray}\label{T1H2}
\sum_{k\geq1}\lambda_i(k)(\rho_i)^k - \sum_{k\geq1}\sum_{j \in V^k_{\cdot \to i}}\lambda_j(k) \geq 0. 
\end{eqnarray}
\end{enumerate}
Then there exists a unique stationary process $(\xi_t)$ taking values ​​on $S$, with trajectory continuous on the right, limit on the left and generator $\LL$ given by (\ref{gerador}).
\end{theo}

The proof of the Theorem \ref{T1} is given in Section \ref{proof13}. It is based on the construction of a perfect simulation procedure for the process and showing that the number of steps of the algorithm is finite almost surely. 

The condition (\ref{T1H2}) is rather restrictive. As an illustration of Theorem \ref{T1} we give the following example of a system with satisfies the assumptions of the theorem. 

\begin{ex}\label{ex1}
Take $\II = \{1,2,\dots\}$ and assume that $|V^k_{\cdot \to i}| \leq c_k$, for all $i \in \II$ and $k\geq1$, where $c_k \in \N$ depends only on $k$. 
Let $s:\N \to \R_+$ be a function such that
\begin{eqnarray*}
\sum_{k\geq0}s(k) < + \infty.
\end{eqnarray*} 
Let $(a_i)_{i \in \II}$ be a non-decreasing sequence of real number with $a_i > 0$ for $i \in \II$. Assume that the sequence $(a_i)$ satisfies the following property
\begin{eqnarray*}
\sup\{\bar{a}_i: i \in \II\}< +\infty,
\end{eqnarray*}
where $\bar{a}_i = \sup\{a_j: j \in \cup_{k\geq 1}V^k_{\cdot \to i}\}$. For each $i \in \II$ define 
\begin{eqnarray*}
\lambda_i(k) = a_is(k) \hbox{\ for\ } k \geq 0.
\end{eqnarray*}

For a fixed $i \in \II$, 
\begin{eqnarray*}
\Lambda_i \leq a_i\sum_{k\geq0}s(k) + \bar{a}_i\sum_{k\geq1}c_ks(k).
\end{eqnarray*}
Thus, the condition (\ref{T1H1}) is satisfied if $\sum_{k\geq1}c_ks(k) < \infty$. Note that $\rho_i = s(0)/\sum_{k\geq0}s(k)$ for all $i \in \II$ and $\sum_{k\geq1}s(k)\rho^k \leq \sum_{k\geq1}s(k)c_k$. Then, if we have the following inequality 
\begin{eqnarray*}
\bar{a}_i\leq a_ir,\hbox{\  where\ } r \in (0,1),
\end{eqnarray*}
the condition (\ref{T1H2}) is satisfied. 
\end{ex}

Now let us consider the process $(\xi^{[F]}_t)$ on $S^{[F]} = \{0,1,\dots\}^F$, where $F$ is a finite subset of $\II$. A natural constraint on a finite set $F \subset \II$ of the transformations $\pi^i_k$ is given by, for $i \in F$ and $k\geq1$, 
\begin{eqnarray*}
\pi^{k,[F]}_i(\xi^{[F]})(u) & = &\left\{\begin{array}{ll}
\xi^{[F]}(u) & \hbox{\ if\ } \xi^{[F]}(i) < k, \\
\xi^{[F]}(u) & \hbox{\ if\ } u\neq i, u \not\in V^k_{i \to \cdot}(F) \hbox{\ and\ } \xi^{[F]}(i) \geq k,\\
\xi^{[F]}(u)+1 & \hbox{\ if\ } u \neq i, u \in V^k_{i \to \cdot}(F) \hbox{\ and\ } \xi^{[F]}(i) \geq k, \\
0 & \hbox{\ if\ } u = i \hbox{\ and\ } \xi^{[F]}(i) \geq k,
\end{array}\right.
\end{eqnarray*}
where $V^k_{i \to \cdot}(F) = V^k_{i \to \cdot}\cap F$. 

For $k = 0$
\begin{eqnarray*}
\pi^{0,^{[F]}}_i(\xi^{[F]})(u) & =& \left\{\begin{array}{ll}
\xi^{[F]}(u) & \hbox{\ if\ } u\neq i, \\
\xi^{[F]}(u)+1 & \hbox{\ if\ } u = i.
\end{array}\right.
\end{eqnarray*}

Then the generator $\LL^{[F]}$ of the process $(\xi_t^{[F]})$ is given by 
\begin{eqnarray}\label{gerador2}
\LL^{[F]} f (\xi^{[F]}) = \sum_{i \in F}\sum_{k\geq 0}\lambda_i(k)[f(\pi^{k,[F]}_i(\xi^{[F]})) - f(\xi^{[F]})].
\end{eqnarray}

The next theorem guarantees the existence of a perfect simulation algorithms based on a coupled construction of the two processes, $(\xi_t)$ and $(\xi^{[F]}_t)$. Moreover, provides an upper bound for the error we make when sampling from a finite set of neurons instead of consider the whole system.

\begin{theo}\label{T2}
Under the assumptions of Theorem \ref{T1}. For any $F$ finite subset of $\II$ and $i \in F$, there exist a coupled perfect sampling algorithm for the pair $(\xi(i),\xi^{[F]}(i))$ satisfying 
\begin{eqnarray}\label{T2B}
P(\xi(i) \neq \xi^{[F]}(i)) \leq \frac{\delta(F)}{1-\alpha}, 
\end{eqnarray}
where 
\begin{eqnarray}
\delta(F) = \sup_{C\subset \II: |C|<\infty}\frac{1}{\sum_{u' \in C}\Lambda_{u'}}\sum_{u \in C}\sum_{k\geq1}\sum_{j \in V^k_{\cdot \to u}}\lambda_j(k)\one\{j \in \bar{F}\}.
\end{eqnarray}
\end{theo}

Note that $\delta(F) \to 0$ when $F \to \II$. The proof of the inequality (\ref{T2B}) will be given in Section \ref{proof2}.

The two theorems show that the process considering an infinite number of neurons is a good reference to phenomena that appear when the number of neurons is very large.

\section{Perfect simulation}\label{PerfectSimulation}

Our goal is to construct a perfect simulation algorithm that generates as an output the value of the membrane potential of any fixed neuron $i \in \II$. The simulation procedure has two stages. In the first stage we determine the set of neurons whose activity influence the membrane potential of neuron $i$ under equilibrium. This stage will be called \textit{backward sketch procedure}. For this, we climbing down from time $0$ back to the past of neuron $i$ until the last time that a Poisson clock rang, with rate $\Lambda_i$. At this time, we choose a pair $(j,\bar{k})$, where $j \in \II$ and $\bar{k} \geq 0$, with a certain probability. If $j=i$ and $\bar{k}\neq 0$, we assign the value $0$ to the value of the membrane potential of neuron $i$. If $j\neq i$ and $\bar{k} \neq 0$, we restart the above procedure for the two neurons $i$ and $j$. Finally, if  $j=i$ and $\bar{k} = 0$ we restart the procedure just with neuron $i$. The procedure stops whenever for all neuron $j$ involved in this backward time evolution one choice of the pair  $(j,\bar{k}\neq0)$ occurs.

When this occurs, we start the second stage, called \textit{forward spin assignment procedure}, in this stage we go back to the future assigning the value of the membrane potential to all neurons visited during the first stage. We know that the value of the membrane potential of all neurons in the end of the first stage is $0$ at the time which $\bar{k}\neq0$ is chosen. Using this information we can determine the value of the membrane potential for all neurons visited in the first stage at each time that the Poisson clock rang up to time $0$.

\subsection{Backward sketch procedure}\label{backward}

Before describing formally the two algorithms described above, let us define the stochastic process which is behind the backward sketch procedure. Our goal is to define, for each neuron $i \in \II$, a process $(C^{(i)}_s)_{s\geq0}$ taking values in the set of finite subsets of $\II$, such that $C^{(i)}_s$ is the set of neurons at time $-s$ whose activity may affect the membrane potential of neuron $i$ at time $t=0$. 

Let $\Sigma =\{\sigma_i^{(j,\bar{k})}: i,j \in \II, \bar{k}\geq0\}$ be a family of transformations in $\PP(\II)$, the set of all finites subsets of $\II$, defined as follows. For any unitary set $\{u\}$, for $j\neq i$ and $\bar{k}\geq 1$

\begin{eqnarray}
\sigma_i^{(j,\bar{k})}(\{u\}) = \left\{\begin{array}{ll}
\{i,j\}, & \hbox{\ if\ } u = i \hbox{\ and\ } j \in V^{\bar{k} }_{\cdot \to i} \\
\{u\}, & \hbox{\ otherwise. } 
\end{array}\right.
\end{eqnarray}

For $j = i$ and $\bar{k} \geq 1$,
\begin{eqnarray}
\sigma_i^{(i,\bar{k})}(\{u\}) = \left\{\begin{array}{ll}
\emptyset, & \hbox{\ if\ } u = i \\
\{u\}, & \hbox{\ otherwise. } 
\end{array}\right.
\end{eqnarray} 

For  $j \in \II$ and $\bar{k} = 0$
\begin{eqnarray}
\sigma_i^{(j,\bar{k})}(\{u\}) = \{u\}, \hbox{\ for all\ } u. 
\end{eqnarray}

For any $C \subset \II$ finite, we define similarly
\begin{eqnarray}
\sigma^{(j,\bar{k})}_i(C) = \bigcup_{u \in C}\sigma^{(j,\bar{k})}_i(\{u\}).
\end{eqnarray}

Let $\nu: \Sigma \to \R_+$ be a function define as follows, for $i \in \II$, 
\begin{eqnarray}
\nu(\sigma_i^{(j,\bar{k})}) = \left\{\begin{array}{ll}
\lambda_j(\bar{k}), & \hbox{\ if\ } j \neq i \hbox{\ and\ } j \in V^{\bar{k}}_{\cdot \to i} \\
\lambda_i(k)(\rho_i)^{\bar{k}}, & \hbox{\ if\ } j = i \hbox{\ and\ } \bar{k}\geq 1 \\
\lambda_i(0) + \sum_{l\geq1}\lambda_{i}(l)(1-(\rho_i)^l), & \hbox{\ if\ } j = i \hbox{\ and\ } \bar{k}=0 \\
0, & \hbox{\ otherwise. \ }
\end{array}\right.
\end{eqnarray}

Let $\{N^\sigma, \sigma \in \Sigma\}$ be a family of independent Poisson point processes with rate $\nu(\sigma)$ respectively.  Let us introduce the Poisson point process $N$ on $\R_+\times\Sigma$ as follows
\begin{eqnarray*}
N(B\times I) = \sum_{\sigma \in I}N^\sigma(B),
\end{eqnarray*} 
where $B \in \B(0,+\infty)$ and $I \subset \Sigma$ finite. For each $i \in \II$, we can define 
\begin{eqnarray}
T^{(i)}_1 = \left\lbrace\begin{array}{ll}
\inf\{t>0: N(]0,t]\times\Sigma_i) > 0\}, & \hbox{\ if\ } \sum_{\sigma \in \Sigma_i}\nu(\sigma) > 0, \\
 + \infty, & \hbox{\ if\ }  \sum_{\sigma \in \Sigma_i}\nu(\sigma) = 0,
\end{array}\right.
\end{eqnarray}
where $\Sigma_i = \{\sigma_i^{(j,\bar{k})}: \bar{k}\geq0, j \in V^{\bar{k}}_{\cdot \to i} \cup \{i\}\}$.
 
By hypothesis (\ref{T1H1}) we have that $\sum_{\sigma \in \Sigma_i}\nu(\sigma) = \Lambda_i < \infty$. Thus, if we exclude the second case, $T^{(i)}_1$ is a exponential random variable with parameter $\Lambda_i$. Moreover, we know that at time $T^{(i)}_1$ a transformation $\Pi^{(i)}_1 \in \Sigma_i$, independent of $T^{(i)}_1$, occurs with probability given by
\begin{eqnarray}
P(\Pi^{(i)}_1 = \sigma) = \left\lbrace \begin{array}{ll}
\frac{\nu(\sigma)}{\Lambda_i}, & \hbox{\ if\ } \sigma \in \Sigma_i, \\
0, & \hbox{\ otherwise.\ }
\end{array}
\right.
\end{eqnarray}

Define 
\begin{eqnarray*}
\tilde{C}^{(i)}_1 = \Pi^{(i)}_1(\{i\}).
\end{eqnarray*} 
 
Let $T$ be a stop time of the process $N$ and $I$ a random subset of $\Sigma$, $\F_T-$measurable, such that $\sum_{\sigma\in I}\nu(\sigma) >0$. Define the random variable $S(T,I)$ by 
\begin{eqnarray*}
S(T,I) = \left\lbrace \begin{array}{ll}
\inf\{t>T: N(]T,t]\times I) > 0\}, & \hbox{\ if\ } T<\infty \hbox{\ and\ } \sum_{\sigma \in I}\nu(\sigma) > 0, \\
\infty, & \hbox{\ otherwise.\ }
\end{array}
\right.
\end{eqnarray*} 

If $S(T,I)$ is finite, we know that $S(T,I)-T$ is an exponential random variable with parameter $\sum_{\sigma \in I}\nu(\sigma)$ and at time $S(T,I)$ the transformation $\Pi(T,I) \in I$ occurs with distribution given by 
\begin{eqnarray*}
P(\Pi(T,I) = \sigma) = \left\lbrace \begin{array}{ll}
\frac{\nu(\sigma)}{\sum_{\sigma \in I}\nu(\sigma)}, & \hbox{\ if\ } \sigma \in I, \\
0, & \hbox{\ otherwise.\ }
\end{array}
\right.
\end{eqnarray*}
 
Now we can define recursively
\begin{eqnarray*}
T^{(i)}_{n+1} &=& S(T_n,\Sigma_{\tilde{C}^{(i)}_n}), \\
\tilde{C}^{(i)}_{n+1} &=& \Pi(T^{(i)}_{n},\Sigma_{\tilde{C}^{(i)}_n})(\tilde{C}^{(i)}_n),
\end{eqnarray*}
 where $\Sigma_{\tilde{C}^{(i)}_n} = \bigcup_{u\in \tilde{C}^{(i)}_n}\Sigma_u$. 
 
The sequence of successive jumps times $(T^{(i)}_n)$ and the Markov chain $(\tilde{C}^{(i)}_n)$ define the backward sketch process starting from $\{i\}$ at time $s=0$, denoted by $(C_s^{(i)})_{s\geq0}$. The infinitesimal generator $\GG$ of the process is given by
\begin{eqnarray}\label{GG}
\GG f(C) = \sum_{\sigma \in \Sigma}\nu(\sigma)[f(\sigma(C)) - f(C)],
\end{eqnarray}
where $f: \PP(\II) \to \R$ is any bounded cylinder function. 

\begin{prop}\label{P1}
\begin{enumerate}
\item Suppose that 
\begin{eqnarray}\label{P1H1}
\sum_{\sigma \in \Sigma_i}\nu(\sigma) < \infty \hbox{\ for all\ }i \in \II.
\end{eqnarray}
Then there exist a probability space $(\Omega,\B,P)$ on which we can construct, for all $i \in \II$, a version of the Markov jump process $(C^{(i)}_s)_{s\geq0}$ with initial condition at time $s=0$, $C^{(i)}_0=\{i\}$, and infinitesimal generator define in (\ref{GG}).
\item Suppose moreover that 
\begin{eqnarray}\label{P1H2}
\sup_{i \in \II}\sum_{\sigma \in \Sigma}\nu(\sigma)(|\sigma(\{i\})|-1) = c < \infty.
\end{eqnarray}
Then, for all $i \in \II$, the process $(C^{(i)}_s)$ does not explode (i.e., $|C^{(i)}_s| \leq e^{cs}$ for any $s\geq0$). 
\end{enumerate}
\end{prop}

\begin{proof}
The first part of the proposition is proved by the above construction. We just need to prove that the process does not explode. For all $n \in \N$ and $i \in \II$, we have the following equality 
\begin{eqnarray}\label{P1P1}
\E(|C^{(i)}_{t\wedge T^{(i)}_n}|) = 1+\sum_{\sigma}\nu(\sigma)\E\left[\int^{t\wedge T^{(i)}_n}_0d(s)(|\sigma(C^{(i)}_s)| - |C^{(i)}_s|)\right].
\end{eqnarray}

By the hypothesis (\ref{P1H2}), we have 
\begin{eqnarray}\label{P1P2}
\sum_{\sigma}\nu(\sigma)(|\sigma(C^{(i)}_s)| - |C^{(i)}_s|) \leq c|C^{(i)}_s|
\end{eqnarray} 

From (\ref{P1P1}) and (\ref{P1P2}), we have that 
\begin{eqnarray*}
\E(|C^{(i)}_{t\wedge T^{(i)}_n}|) \leq 1 + c\E\left[\int^t_0|C^{(i)}_{s\wedge T^{(i)}_n}|ds\right].
\end{eqnarray*} 

By the inequality of Gronwall, we have 
\begin{eqnarray}\label{P1P3}
\E(|C^{(i)}_{t\wedge T^{(i)}_n}|) \leq e^{ct} \hbox{\ for all\ } n \in \N. 
\end{eqnarray}

Note that the above inequality holds for all $n \in \N$. Thus, taking the limit of equation (\ref{P1P3}), we have
\begin{eqnarray*}
\E(|C^{(i)}_t|) \leq e^{ct} \hbox{\ for\ } t \in \R_+.
\end{eqnarray*}
\end{proof}

In the same way, we can define the process $(C^{(i),[F]}_s)_{s\geq0}$, for any $F$ finite subset of $\II$ and $i \in \II$, define by the  infinitesimal generator $\GG^{[F]}$ given by
\begin{eqnarray}\label{GGF}
\GG^{[F]} f(C) = \sum_{\sigma \in \Sigma(F)}\nu(\sigma)[f(\sigma(C)) - f(C)],
\end{eqnarray}
where $f: \PP(F) \to \R$ is any bounded cylinder function and $\Sigma(F) = \{\sigma_i^{(j,\bar{k},F)} \in \Sigma : i \in F, j \in F, \bar{k}\geq0\}$. The transformations $\sigma^{(j,\bar{k},F)}_i: \PP(F) \to \PP(F)$ are equal to the transformations $\sigma^{(j,\bar{k})}$ restricted to the set $F$. 

For the perfect simulation of the coupled process, we need to construct a coupling procedure for the processes $(C^{(i)}_s)$ and $(C^{(i),[F]}_s)$ such that 
\begin{eqnarray}\label{P3I1}
C^{(i),[F]}_s \subseteq C^{(i)}_s, \hbox{\ for all\ } s \in \R_+.
\end{eqnarray}
The next result guarantees the existence of such coupling.  

\begin{prop}
For any $F \subset \II$ finite and $i \in F$, exist a coupling between the process $(C^{(i)}_s)$ and $(C^{(i),[F]}_s)$ such that the equation (\ref{P3I1}) is satisfied. 
\end{prop}

\begin{proof}
Let $(T^{(i)}_n)_{n\geq0}$ the sequence of successive jumps times define above. For $F$ finite subset of $\II$ and $i \in F$, we will build a coupling between the process $(C^{(i)}_s)$ and $(C^{(i),[F]}_s)$ as follow:
\begin{enumerate}
\item For $t \in [0,T^{(i)}_1[$, define $C_t^{(i)} = C_t^{(i),[F]} = \{i\}$.
\item For $n\geq 1$:
\begin{enumerate}
\item Given $C_{T^{(i)}_{n-1}}^{(i)}$, we choose a transformation $\Pi_n$ with probability distribution given by 
\begin{eqnarray*}
P(\Pi_n = \sigma) = \left\lbrace\begin{array}{ll}
\frac{\nu(\sigma)}{\sum_{\sigma' \in \Sigma_n}\nu(\sigma')}, & \hbox{\ if\ } \sigma \in \Sigma_n, \\
0, & \hbox{\ otherwise,\ }
\end{array}\right.
\end{eqnarray*} 
where $\Sigma_n = \{\sigma_u^{(j,k)}: u \in C^{(i)}_{T^{(i)}_{n-1}}, k\geq0, j \in V^k_{\cdot \to u}\cup\{u\}\}$. 
\item Define for $T_n^{(i)} \leq t < T^{(i)}_{n+1}$
\begin{eqnarray*}
C_t^{(i),[F]} = \left\lbrace\begin{array}{ll}
\Pi_n(C_{T^{(i)}_{n-1}}^{(i),[F]}), & \hbox{\ if\ } \Pi_n \in \Sigma_n(F), \\
C_{T^{(i)}_{n-1}}^{(i),[F]} & \hbox{\ otherwise,\ }
\end{array}\right.
\end{eqnarray*}
where $\Sigma_n(F) = \{\sigma_u^{(j,k)}: u \in C_{T^{(i)}_{n-1}}^{(i),[F]}, k\geq0, j \in V^k_{\cdot \to u}(F)\cup \{u\}\}$.
\end{enumerate} 
\end{enumerate}  
Thus, by the above construction, the inequality (\ref{P3I1}) holds. 
\end{proof}

\subsection{Algorithms}\label{algorithm}

Now we can introduce the algorithm that generates as an output the value of the pair $(\xi(i), \xi^{[F]}(i))$ of the coupled process, for any $F \subset \II$ finite and $i \in F$ at time $t=0$. The following variables will be used.

\begin{itemize}
\item $N$ is an auxiliary variable taking values in the set of non-negative integers $\{0,1,2,\dots\}$.
\item $N^{(i)}_{STOP}$ and $N^{(i),[F]}_{STOP}$ are counters taking values in the set of non-negative integers $\{0,1,2,\dots\}$.
\item $B$ and $B^{[F]}$ are arrays of elements of $\II\times\{0,1,2,\dots\}$.
\item $C$ and $C^{[F]}$ are variables taking values in the set of finite subsets of $\II$.
\item For $k\geq1$, $V(k)$ is a variable taking values in the set of finite subsets of $\II$.
\item $\bar{V}$ is a variable taking values in the set of finite subsets of $\II$. 
\item $p$ is a vector of elements of $\{0,1,2,\dots\}$.
\item $\xi_0(i)$ is a variable taking values in the set of non-negative integers $\{0,1,2,\dots\}$. 
\end{itemize}

We will present the backward sketch procedure for the coupled perfect simulation algorithm. The backward sketch procedure for the process $(\xi_t)$ can be immediately deduce ignoring the steps 20 to 30.

\vspace{5pt}

\begin{tabular}{l}
\textbf{Algorithm 1} Backward sketch procedure   \\
\hline
\end{tabular}
\begin{enumerate}
\item \textbf{Input:} $F \subset \II$, $i \in F$, $\left\lbrace V^k_{\cdot \to j}, j\in \II, k\geq 1\right\rbrace$

\textbf{Output:} $N^{(i)}_{STOP}$, $N^{(i),[F]}_{STOP}$, $B$ and $B^{[F]}$
\item $N \leftarrow 0, N^{[F]} \leftarrow 0, C \leftarrow \{i\}$ and $C^{[F]} \leftarrow \{i\}$
\item \textbf{While} $C \neq \emptyset$ \textbf{do}
\item \hspace{0.5cm}  $N \leftarrow N+1$
\item \hspace{0.5cm} For each $k \geq 1$, define $V(k) \leftarrow \bigcup_{u \in C}V^k_{\cdot \to u} \cup \{u\}$
\item \hspace{0.5cm} Choose $J$ and $K$ randomly in $\II\times\{0,1,\dots\}$ according to the probability distribution  $$P(J=j,K=k)=\frac{\lambda_j(k)\one\{j \in V(k)\}}{\sum_{u \in C}\Lambda_u}$$
\item \hspace{0.5cm} \textbf{If} $J \in C$ and $K \neq 0$ \textbf{then}
\item \hspace{1cm} $M \leftarrow 0$
\item \hspace{1cm} \textbf{While} $M < K$ \textbf{do}
\item \hspace{1.5cm} Choose an integer $K'$ randomly in $\{0,1,\dots\}$ according to the probability distribution $$P(K'=k)=\frac{\lambda_J(k)}{\sum_{l\geq0}\lambda_J(l)}$$
\item \hspace{1.5cm} \textbf{If} $K' = 0$ \textbf{then} $M \leftarrow M+1$
\item \hspace{1.5cm} \textbf{Else}, $M \leftarrow K+1$
\item \hspace{1.5cm} \textbf{End If}
\item \hspace{1cm} \textbf{End While}
\item \hspace{1cm} \textbf{If} $M = K$ \textbf{then} $C \leftarrow C/\{J\}$
\item \hspace{1cm} \textbf{End If}
\item \hspace{0.5cm} \textbf{Else} $C \leftarrow C \cup \{J\}$
\item \hspace{0.5cm} \textbf{End If}
\item \hspace{0.5cm} $B(N) \leftarrow (J, K)$
\item \hspace{0.5cm} \textbf{If} $C^{[F]} \neq \emptyset$
\item \hspace{1cm} \textbf{If} $J \in F$ \textbf{then}
\item \hspace{1.5cm} $N^{[F]} \leftarrow N^{[F]}+1$
\item \hspace{1.5cm} $B^{[F]}(N^{[F]}) \leftarrow (J,K)$
\item \hspace{1.5cm} \textbf{If} $J \in C^{[F]}$ \textbf{then}
\item \hspace{2cm} \textbf{If} $M = K$ \textbf{then} $C^{[F]} \leftarrow C^{[F]}/\{J\}$
\item \hspace{2cm} \textbf{End If}
\item \hspace{1.5cm} \textbf{End If} 
\item \hspace{1cm} \textbf{Else} $C^{[F]} \leftarrow C^{[F]}\cup\{J\}$
\item \hspace{1cm} \textbf{End If}
\item \hspace{0.5cm} \textbf{End If}
\item \textbf{End While}
\item $N^{(i),{[F]}}_{STOP} \leftarrow N$ and $N^{(i)}_{STOP} \leftarrow N^{[F]}$
\item \textbf{Return:} $N^{(i)}_{STOP}$, $N^{(i),[F]}_{STOP}$, $B$ and $B^{[F]}$
\end{enumerate}

In the end of the first stage, the algorithm returns $N^{(i)}_{STOP} (N^{(i),[F]}_{STOP})$ the number of steps we need to go back in the past to determine the value of $\xi(i) (\xi^{[F]}(i))$. As well as $B (B^{[F]})$ the set of ancestors of neuron $i$. 

If $B = B^{[F]}$, we know that all the neurons that were visited in the first stage belongs to the set $F$. Then we use one time the forward spin assignment procedure to determine $\xi(i) = \xi^{[F]}(i)$. If $B \neq B^{[F]}$, then we use the procedure twice independently, with input $B$, $N^{(i)}_{STOP}$ to determine $\xi(i)$ and $B^{[F]}$, $N^{(i),[F]}_{STOP}$ to determine $\xi^{[F]}(i)$. 

The second stage uses the output of the first stage to determine the values of $\xi(i)$ and $\xi^{[F]}(i)$ and the procedure is given by the following algorithm. 

\vspace{5pt}

\begin{tabular}{l}
\textbf{Algorithm 2} Forward spin assignment procedure   \\
\hline
\end{tabular}
\begin{enumerate}
\item \textbf{Input:} $i \in \II$, $N^{(i)}_{STOP}$, $B$, $\left\lbrace V^k_{j \to \cdot}, j \in \II, k \geq 1\right\rbrace$

\textbf{Output:} $\{\xi(i)\}$
\item $N \leftarrow N^{(i)}_{STOP}, \bar{V} \leftarrow \emptyset, p \leftarrow \emptyset, X_t(i) \leftarrow 0$
\item \textbf{For} $n \leftarrow N$ \textbf{to} 1 \textbf{do}
\item \hspace{0.5cm} $J \leftarrow B[n,1]$
\item \hspace{0.5cm} $K \leftarrow B[n,2]$
\item \hspace{0.5cm} \textbf{If} $J \not\in \bar{V}$ \textbf{then}
\item \hspace{1cm} $\bar{V} \leftarrow \bar{V}\bigcup J$
\item \hspace{1cm} $p[J] \leftarrow 0$	
\item \hspace{0.5cm}  \textbf{Else}
\item \hspace{1cm} \textbf{If} $K = 0$ \textbf{then} $p[J] \leftarrow p[J]+1$
\item \hspace{1cm} \textbf{End If}
\item \hspace{1cm} \textbf{If} $k \leq p[J]$ \textbf{then}
\item \hspace{1.5cm} $p[J] \leftarrow 0$
\item \hspace{1.5cm} \textbf{For} $m \in V^K_{J\to\cdot}$ \textbf{do}
\item \hspace{2cm} \textbf{If} $m \in \bar{V}/\{J\}$ \textbf{then} $p[m] \leftarrow p[m]+1$
\item \hspace{2cm} \textbf{End If}
\item \hspace{1.5cm} \textbf{End For}
\item \hspace{1cm} \textbf{End If}
\item \hspace{0.5cm} \textbf{End If}
\item \textbf{End For}
\item \textbf{Return} $\xi(i) \leftarrow p[i]$
\end{enumerate}

Let 
\begin{eqnarray*}
T^{(i)}_{STOP} = \inf\{s>0: C_s^{(i)} = \emptyset\}
\end{eqnarray*}
be the time in which the process $C^{(i)}_s$ reaches the empty set. Note that if $T^{(i)}_{STOP} < \infty$, then $C^{(i)}_s = \emptyset$ for all $s \geq T^{(i)}_{STOP}$. 

We also can define 
\begin{eqnarray}\label{ST}
N^{(i)}_{STOP} = \inf\{n>0: \tilde{C}^{(i)}_n = \emptyset\}
\end{eqnarray}
the number of steps that we have to go back to determine the set of ancestors. The definition in (\ref{ST}) is equivalent of the $N^{(i)}_{STOP}$ defined in the Algorithm 1. 
 
For the algorithm to be successful, we need to show that the number of steps $N^{(i)}_{STOP}$ is finite. This is the content of the next theorem.

\begin{theo}\label{T3}
Under the codition 
\begin{eqnarray}\label{T3H1}
\sup_{i \in \II}\sum_{\sigma\in \Sigma_i}\frac{\nu(\sigma)|\sigma(\{i\})|}{\Lambda_i} = \alpha \leq 1,
\end{eqnarray}
we have for any $i \in \II$,
\begin{eqnarray}\label{T3D1}
P(N^{(i)}_{STOP} > n) \leq \alpha^n
\end{eqnarray}
and
\begin{eqnarray}\label{T3D2}
\E(N^{(i)}_{STOP}) \leq \frac{1}{1-\alpha}.
\end{eqnarray}
\end{theo}

The proof of the Theorem \ref{T3} will be given at Section \ref{proof13}. Note that the condition (\ref{T1H2}) implies the condition (\ref{T3H1}). Let us start by proving Theorem \ref{T3}.

\section{Proof Theorem \ref{T1} and \ref{T3}}\label{proof13}

This section is devoted to the proof of Theorem \ref{T1} and Theorem \ref{T3}. To proof the Theorem \ref{T1} we use the results given by the Theorem \ref{T3}. 

\subsection{Proof Theorem \ref{T3}}
Let us start by proving the first inequality (\ref{T3D1}). Let $Z_n^{(i)} = |\tilde{C}^{(i)}_n|$ be the cardinality of the set $\tilde{C}^{(i)}_n$ after $n$ steps of the Algorithm 1.    

For $i \in \II$, we have 
\begin{eqnarray}\label{T3P1}
P\left\lbrace N^{(i)}_{STOP} > n\right\rbrace &\leq & P\left\lbrace |\tilde{C}^{(i)}_n| \geq 1\right\rbrace \nonumber \\
& = & \E(Z_n)
\end{eqnarray}

Given $\tilde{C}^{(i)}_n$, the expected value of $Z_{n+1}$ is given by
\begin{eqnarray*}
\E(Z_{n+1}|\tilde{C}^{(i)}_n) = \sum_{u\in \tilde{C}^{(i)}_n}\sum_{\sigma\in \Sigma_u}\frac{\nu(\sigma)}{\sum_{\sigma'\in\Sigma_u}\nu(\sigma')}|\nu(\sigma)|.
\end{eqnarray*}

By the condition (\ref{T3H1})
\begin{eqnarray}\label{T3P2}
\E(Z_{n+1}) &\leq & \alpha\E(Z_n) \nonumber \\
&\leq & \alpha^{n+1}.
\end{eqnarray}

Then by (\ref{T3P1}) and (\ref{T3P2}), we conclude the proof of the first part   
\begin{eqnarray*}
P\left\lbrace N^{(i)}_{STOP} > n\right\rbrace \leq \alpha^n. 
\end{eqnarray*}

For the second inequality, let $\{Y^{(i)}_n; n\in \N, i \in \II\}$ a family of independent random variables taking values in $\{-1,0,1\}$, independent of the process $(\tilde{C}_n^{(i)})$, such that 
\begin{eqnarray*}
P(Y^{(i)}_n = -1) &=& \frac{\sum_{k\geq1}\lambda_i(k)(\rho_i)^k}{\Lambda_i}, \\
P(Y^{(i)}_n = 0) &=& \frac{\lambda_i(0)+\sum_{k\geq1}\lambda_i(k)(1-(\rho_i)^k)}{\Lambda_i},  \\
P(Y^{(i)}_n = 1) &=& \frac{\sum_{k\geq1}\sum_{j \in V^k_{\cdot \to i}}\lambda_j(k)}{\Lambda_i}.
\end{eqnarray*} 

By condition \ref{T3H1}, we have  
\begin{eqnarray}
\sup_{i \in \II}\E(Y^{(i)}_n) = \alpha - 1 \leq 0.
\end{eqnarray}

For any $i \in \II$, let us call $J^{(i)}_n$ the index of the neuron whose Possion clock rings at the $n-$th jump of the process $(C^{(i)}_s)$.  Recall that $J_n^{(i)}$ are conditionally independet given the sequence $(\tilde{C}^{(i)}_n)_n$ such that 
\begin{eqnarray*}
P(J^{(i)}_n = j| \tilde{C}^{(i)}_{n-1}) = \frac{\sum_{k\geq0}\lambda_j(k)}{\sum_{u \in \tilde{C}^{(i)}_{n-1}}\Lambda_u}, \hbox{\ if\ } j \in \tilde{C}^{(i)}_{n-1}.
\end{eqnarray*}

Let $(S_n^{(i)})_{n\geq0}$ be a process taking values in $\Z$, with initial condition $S^{(i)}_0 = 0$, such that 
\begin{eqnarray*}
S^{(i)}_n = \sum_{m=1}^n Y_m^{J^{(i)}_m}, \hbox{\ for\ } n \geq 1.  
\end{eqnarray*}

Note that by construction, $M_n = S^{(i)}_n + (1-\alpha)n$ is a super-martingale. Then a very rough upper bound is 
\begin{eqnarray*}
Z_n^{(i)} \leq  1 + S_n^{(i)} \hbox{\ as long as\ } n \leq V_{STOP},
\end{eqnarray*}
where $V_{STOP} = \min\{n\geq 1: S_n = -1\}$. Then , by construction 
\begin{eqnarray*}
N^{(i)}_{STOP} \leq V^{(i)}_{STOP}.
\end{eqnarray*}

By the stopping rule for super-martingales, we have that 
\begin{eqnarray*}
\E(S^{(i)}_{V^{(i)}_{STOP}\wedge N}) +(1-\alpha)\E(V^{(i)}_{STOP}\wedge N) \leq 0.
\end{eqnarray*}

But notice that 
\begin{eqnarray*}
\E(S^{(i)}_{V^{(i)}_{STOP}\wedge N}) = - P(V^{(i)}_{STOP} \leq N) + \E(S^{(i)}_N; V^{(i)}_{STOP} > N).
\end{eqnarray*}

On $\{V^{(i)}_{STOP} > N\}$, $S^{(i)}_N \geq 0$, hence we have that $\E(S^{(i)}_{V^{(i)}_{STOP}\wedge N}) \geq -P(V^{(i)}_{STOP} \leq N) $. We conclude that 
\begin{eqnarray*}
\E(V^{(i)}_{STOP}\wedge N) \leq \frac{1}{1-\alpha}P(V^{(i)}_{STOP} \leq N).
\end{eqnarray*}

Now, letting $N\to \infty$, we get 
\begin{eqnarray*}
\E(V^{(i)}_{STOP}) \leq \frac{1}{1-\alpha} \hbox{\ and therefore\ } \E(N^{(i)}_{STOP}) \leq \frac{1}{1-\alpha}. 
\end{eqnarray*}

\subsection{Proof of Theorem \ref{T1}}

It can be shown that if we can construct a perfect simulation algorithm and such algorithm is successful, then the process exist and is unique by construction. The algorithm is successful if the number of steps $N^{(i)}_{STOP}$ is finite almost surely. 

The construction of the algorithm was given in Section \ref{PerfectSimulation}. Here we are going to proof that the number of steps $N^{(i)}_{STOP}$ is finite almost surely. 

By Theorem \ref{T3}, we have that 
\begin{eqnarray*}
P(N^{(i)}_{STOP} > n) \leq \alpha^n,
\end{eqnarray*}  
where $\alpha \in (0,1)$. This implie that 
\begin{eqnarray*}
\sum_{n=0}^{\infty}P(N^{(i)}_{STOP} > n) \leq \frac{1}{1-\alpha} < \infty
\end{eqnarray*}

Thus, by Borel Cantelli, we have that 
\begin{eqnarray*}
P(N^{(i)}_{STOP} > n) = 0,
\end{eqnarray*} 
for all $n$ sufficiently large. Therefore, the Algorithm 1 is successful and the exist a unique stationary process $(\xi_t)$ taking values in $S$ with generator given by (\ref{gerador}). 

\section{Proof of Theorem \ref{T2}}\label{proof2}

We have to show that the coupled perfect sampling Algorithm 1 achieves the bound (\ref{T2B}). Let $T_F^{(i)}$ be the first time that a neuron $j$ that does not belong to the set $F$ is chosen in the draws that define the set of ancestors, given by
\begin{eqnarray}
T^{(i)}_F = \inf\{t>0: C^{(i)}_t\cup \bar{F} \neq \emptyset\}.
\end{eqnarray}

Recall that in order to construct $\xi$ and $\xi^{[F]}$, we use the same Poisson point processes for all $j \in F$. Thus we have that 
\begin{eqnarray}
P(\xi(i) \neq \xi^{[F]}(i)) \leq P(T^{(i)}_F \leq T^{(i)}_{STOP}).
\end{eqnarray}

By Lemma given below we conclude the proof of Theorem \ref{T3}.

\begin{lem}
For any $F$ finite subset of $\II$ and $i \in F$, we have that 
\begin{eqnarray}
P(T^{(i)}_F \leq T^{(i)}_{STOP}) \leq \frac{\delta(F)}{1-\alpha}.
\end{eqnarray}
\end{lem}

\begin{proof}
For $F$ and $C$ finite subsets of $\II$ define the event 
\begin{eqnarray*}
A_n(F,C) = \left\lbrace \Pi(T^{(i)}_{n},\Sigma_C) = \sigma, \hbox{\ for some\ } \sigma \in \bar{\Sigma}_C(F)\right\rbrace,
\end{eqnarray*}
where $\bar{\Sigma}_C(F) = \{\sigma_u^{(j,k)}: u \in C, k\geq1, j\in V^k_{\cdot \to u}\cap \bar{F}\}$ and $\Pi$ is a random variable taking values in $\Sigma$ define in Section \ref{PerfectSimulation}. Thus, 
\begin{eqnarray*}
\{T^{(i)}_F \leq T^{(i)}_{STOP}\} = \bigcup_{n\geq1}\bigcup_{C \in \PP(F)}\{A_n(F,\tilde{C}_{n-1}^{(i)}), \tilde{C}_{n-1}^{(i)}=C, N^{(i)}_{STOP} > n-1\}
\end{eqnarray*}
and 
\begin{eqnarray*}
P(T_F^{(i)} < T^{(i)}_{STOP}) &\leq & \sum_{n\geq1}\sum_{C \in \PP(\II)}P(A_n(F,\Sigma_{\tilde{C}_{n-1}^{(i)}})|\tilde{C}_{n-1}^{(i)}=C, N^{(i)}_{STOP} > n-1) \\
& &  \qquad \qquad \qquad \qquad \qquad \  \times P(\tilde{C}_{n-1}^{(i)}=C, N^{(i)}_{STOP} > n-1) \\
&=& \sum_{n\geq1}\E(P(A_n(F,\Sigma{\tilde{C}_{n-1}^{(i)}})|\tilde{C}_{n-1}^{(i)}=C, N^{(i)}_{STOP} > n-1)).
\end{eqnarray*}

Given $\tilde{C}^{(i)}_{n-1}$, we have that 
\begin{eqnarray*}
P(\Pi(T^{(i)}_n,\Sigma{\tilde{C}^{(i)}_{n-1}}) = \sigma^{(j,k)}_u) = \frac{\lambda_j(k)\one\{j\in \bar{F}\}}{\sum_{u'\in \tilde{C}^{(i)}_{n-1}}\Lambda_{u'}},
\end{eqnarray*}
for $u \in \tilde{C}^{(i)}_{n-1}$, $k\geq1$ and $j \in V^k_{\cdot \to u}$. Hence, 
\begin{eqnarray*}
P(A_n(F,\Sigma{\tilde{C}_{n-1}^{(i)}})|\tilde{C}_{n-1}^{(i)}=C, N^{(i)}_{STOP} > n-1) = q(F,C)\one\{N^{(i)}_{STOP} > n-1\},
\end{eqnarray*}
where 
\begin{eqnarray*}
q(F,C) = \frac{1}{\sum_{u'\in C}\Lambda_{u'}}\sum_{u\in C}\sum_{k\geq1}\sum_{j\in V^k_{\cdot\to u}}\lambda_j(k)\one\{j \in \bar{F}\}
\end{eqnarray*}

Note that, fo any $C$ finite subset of $\II$, we have that $q(F,C) \in [0,1]$. Define, for $F$ finite subset of $\II$,
\begin{eqnarray}
\delta(F) = \sup\{q(F,C): C \hbox{\ finite subset of\ } \II\}.
\end{eqnarray}

Thus
\begin{eqnarray*}
P(T_F^{(i)} < T^{(i)}_{STOP}) &\leq & \sum_{n\geq1}\E(\delta(F)\one\{N^{(i)}_{STOP} > n-1\}) \\
& = & \delta(F)\sum_{n\geq1}P(N^{(i)}_{STOP} > n-1) \\
& = & \delta(F)\E(N^{(i)}_{STOP}) 
\end{eqnarray*}

Concluding the proof, since $\E(N^{(i)}_{STOP})=\frac{1}{1-\alpha}$ by Theorem \ref{T3}. 
\end{proof}

\section*{Acknowledgments}
This work was produced as part of the activities of FAPESP  Research, Innovation and Dissemination Center for Neuromathematics (grant 2013/ 07699-0, S. Paulo Research Foundation) and CAPES/Nuffic project (grant 038/12). Many thanks to Antonio Galves by all the guidance throughout this work. I would like to thank Antonio Carlos Roque da Silva Filho for helpfull discussions.

\bibliography{bibliografia}
\nocite{fernandez-2002}
\nocite{ferrari-2002}
\nocite{gallo-2011}
\nocite{galves-1977}
\nocite{galves-2010}
\nocite{galves-eva-2010}
\nocite{galves-nancy-2013}
\nocite{propp-96}
\nocite{spitzer-1970}
\nocite{liggett-2011}
\nocite{bruno-2011}

\end{document}